\documentclass[reqno,11pt,a4paper]{amsart}
\allowdisplaybreaks[4]
\usepackage[foot]{amsaddr}
\usepackage{amssymb}
\usepackage{amsmath}
\usepackage[all]{xy}
\usepackage{color}
\usepackage[colorlinks,linkcolor=blue,citecolor=red]{hyperref}

\DeclareMathOperator{\liminv}{liminv}
\newcommand*{\eval}[1]{\left.#1\right|}

\newcommand*{\abs}[1]{\lvert{#1}\rvert}
\def\d{\mathrm{d}}
\newcommand*{\s}[1]{Q_{[#1]}}

\theoremstyle{theorem}

\newtheorem{lemma}{Lemma}
\newtheorem{corollary}{Corollary}
\newtheorem{proposition}{Proposition}

\theoremstyle{definition}

\theoremstyle{remark}
\newtheorem*{remark}{Remark}

\newcommand{\ds}{\displaystyle}
\newcommand{\p}{\partial}
\newcommand{\cprime}{\/{\mathsurround=0pt$'$}}
\newcommand{\const}{\mathrm{const}}

\usepackage{mathrsfs}
\let\mathcal\mathscr
\usepackage[mathscr]{eucal}

\let\kappa=\varkappa
\let\phi=\varphi
\let\rho=\varrho

\begin{document}

\title[Integrability of S-deformable surfaces]{Integrability of S-deformable
  surfaces:\\ conservation laws, Hamiltonian structures and more}
\author{I.S.~Krasil{\cprime}shchik}\address[IK]{Independent University of Moscow,
  B. Vlasevsky 11, 119002 Moscow, Russia \& Mathematical
  Institute, Silesian University in Opava, Na Rybn\'{\i}\v{c}ku~1, 746 01
  Opava, Czech Republic}
\email{josephkra@gmail.com} \author{A.~Sergyeyev$^*$}\address[AS]{Mathematical
  Institute, Silesian University in Opava, Na Rybn\'{\i}\v{c}ku 1, 746 01
  Opava, Czech Republic} \email{Artur.Sergyeyev@math.slu.cz}

\thanks{$^*$Corresponding author}
  \date{}

  \begin{abstract}
    We present infinitely many nonlocal conservation laws, a pair of
    compatible local Hamiltonian structures and a recursion operator for the
    equations describing surfaces in three-dimensional space that admit
    nontrivial deformations which preserve both principal directions and
    principal curvatures (or, equivalently, the shape operator).
\end{abstract}
\keywords{Nonlocal conservation laws, Hamiltonian structures, recursion
  operator, shape operator, Gauss--Codazzi equations, $S$-deformable surfaces}
\subjclass[2010]{37K05, 37K10, 53A05}
\maketitle
\tableofcontents
\newpage

\section{Introduction}
\label{sec:introduction}
The class of surfaces in three-dimensional space that admit nontrivial
deformations which simultaneously preserve principal directions and principal
curvatures (or, equivalently, the shape operator, also known as the Weingarten
operator) has \cite{evf} a long and distinguished history: it was studied
already by Finikoff and Gambier \cite{f, fg} more than 80 years ago but the
investigation of preservation of principal directions and principal curvatures
dates back to Bonnet \cite{b1, b2}. For the sake of brevity we shall refer to
the surfaces from the class in question as to the $S$-{\em deformable} {\em
  surfaces}.

Ferapontov \cite{evf} has established integrability of the corresponding
Gauss--Codazzi equations (\ref{eq:2}) by presenting the associated Lax pair
with a nonremovable spectral parameter; cf.\ also \cite{km} and references
therein for the general study of integrability of the Gauss--Codazzi
equations.  \looseness=-1

A natural next step in the study of the equations in question, which we
rewrite in the form (\ref{sys-eta}), is to explore their geometric structures
naturally related to integrability: symmetries, conservation laws, Hamiltonian
structures and recursion operators.

In what follows we implement this program. Namely, after recalling the
explicit form of the equations under study and of their Lax pair in
Section~\ref{p} we proceed to construct an infinite sequence of nontrivial
nonlocal conservation laws for (\ref{sys-eta}) in
Section~\ref{sec:infin-hier-nonl}, and a recursion operator along with a pair
of compatible local Hamiltonian structures in Sections~\ref{sec-ro} and
\ref{sec-ho}.

\section{Preliminaries}\label{p}
Consider the system \cite{evf} of the Gauss--Codazzi equations describing the
$S$-deformable surfaces
\begin{gather}\nonumber
  \partial_1H_2=\beta_{12}H_1,\quad \partial_2H_1=\beta_{21}H_2,\\
  \partial_1\beta_{12}+\partial_2\beta_{21}=0,\label{eq:2}\\\nonumber
  \eta_1\partial_1\beta_{12} + \eta_2\partial_2\beta_{21} +
  \frac{1}{2}\eta_{1}'\beta_{12} + \frac{1}{2}\eta_2'\beta_{21} + H_1H_2 = 0,
\end{gather}
where $\eta_1=\eta_1(x)$ and $\eta_2=\eta_2(y)$ are arbitrary smooth
functions.

Upon expressing $\beta_{ij}$ via $H_k$ we rewrite this system in the form
\begin{equation}\label{sys-eta}
  \begin{array}{rcl}
    u_{yy} &=& \displaystyle\frac{\frac{1}{2}\eta_1' v^2
               v_x+\frac{1}{2}\eta_2' u v
               u_y +(\eta_1-\eta_2) u u_y v_y+u^2 v^3} {u v
               (\eta_1-\eta_2)},\\[5mm]
    v_{xx}&=& \displaystyle\frac{\frac{1}{2}\eta_1' u v v_x
              +\frac{1}{2}\eta_2' u^2 u_y+ (\eta_2-\eta_1)
              v v_x u_x+u^3 v^2}{u v(\eta_2-\eta_1)},
  \end{array}
\end{equation}
where~$u=H_1$, $v=H_2$. Its (complex) $\mathfrak{sl}_2$-valued zero-curvature
representation reads \cite{evf}
\begin{align}
  \label{eq:3}
  \begin{pmatrix}
    \psi_1\\
    \psi_2
  \end{pmatrix}_x&=
                   \frac{1}{2\sqrt{\lambda+\eta_1}}\begin{pmatrix}
                     i\sqrt{\lambda+\eta_2}\,\cdot\dfrac{u_y}{v}&u\\
                     -u&-i\sqrt{\lambda+\eta_2}\,\cdot\dfrac{u_y}{v}
                   \end{pmatrix}
                         \begin{pmatrix}
                           \psi_1\\
                           \psi_2
                         \end{pmatrix},\\\nonumber
  \begin{pmatrix}
    \psi_1\\
    \psi_2
  \end{pmatrix}_y&=
                   \frac{i}{2\sqrt{\lambda+\eta_2}}\begin{pmatrix}
                     -\sqrt{\lambda+\eta_1}\,\cdot\dfrac{v_x}{u}&v\\
                     v&\phantom{-i}\sqrt{\lambda+\eta_1}\,\cdot\dfrac{v_x}{u}
                   \end{pmatrix}
                        \begin{pmatrix}
                          \psi_1\\
                          \psi_2
                        \end{pmatrix},
\end{align}
where $\lambda\in\mathbb{R}$.

For any zero-curvature representation of the form
\begin{align*}
  \begin{pmatrix}
    \psi^1\\\psi^2
  \end{pmatrix}_x&=
                   \begin{pmatrix}
                     A_1^1&A_2^1\\
                     A_1^2&A_2^2
                   \end{pmatrix}
                            \begin{pmatrix}
                              \psi_1\\\psi_2
                            \end{pmatrix},\\
  \begin{pmatrix}
    \psi^1\\\psi^2
  \end{pmatrix}_y&=
                   \begin{pmatrix}
                     B_1^1&B_2^1\\
                     B_1^2&B_2^2
                   \end{pmatrix}
                            \begin{pmatrix}
                              \psi^1\\\psi^2
                            \end{pmatrix}
\end{align*}
we can (cf.\ e.g.\ \cite{m} and references therein) consider the associated
Riccati covering
\begin{align*}
  w_x&=A_2^1+(A_1^1-A_2^2)w-A_1^2w^2,\\
  w_y&=B_2^1+(B_1^1-B_2^2)w-B_1^2w^2
\end{align*}
where~$w=\psi^1/\psi^2$; see e.g.\ \cite{k-v} and references therein for more
details on (differential) coverings.

In particular, for~\eqref{eq:3} the Riccati covering is
\begin{align}
  \label{eq:4}
  w_x&=i\sqrt{\frac{\lambda+\eta_2}{\lambda+\eta_1}}\cdot\frac{u_y}{v}w +
       \frac{u}{2\sqrt{\lambda+\eta_1}}(1+w^2),\\\nonumber
  w_y&=-i\sqrt{\frac{\lambda+\eta_1}{\lambda+\eta_2}}\cdot\frac{v_x}{u}w +
       \frac{iv}{2\sqrt{\lambda+\eta_2}}(1-w^2).
\end{align}
By changing the parametrization~$\lambda=1/\mu^2$, $\mu>0$, we
transform~\eqref{eq:4} to
\begin{align}
  \label{eq:5}
  w_x&=i\sqrt{\frac{1+\eta_2\mu^2}{1+\eta_1\mu^2}}\cdot\frac{u_y}{v}w +
       \frac{u\mu}{2\sqrt{1+\eta_1\mu^2}}(1+w^2),\\\nonumber
  w_y&=-i\sqrt{\frac{1+\eta_1\mu^2}{1+\eta_2\mu^2}}\cdot\frac{v_x}{u}w +
       \frac{iv\mu}{2\sqrt{1+\eta_2\mu^2}}(1-w^2).
\end{align}

\section{Infinite hierarchy of nonlocal conservation laws}
\label{sec:infin-hier-nonl}

Consider the formal Taylor expansions
\begin{align*}
  w&=w_0+w_1\mu+w_2\mu^2+w_3\mu^3\dots,\\
  w^2&=w_0^2+2w_0w_1\mu+(2w_0w_2+w_1^2)\mu^2+2(w_0w_3+w_1w_2)\mu^3+\dots,\\
  \frac{1}{\sqrt{1+\eta_1\mu^2}}&=\alpha_0^1+\alpha_2^1\mu^2+ \dots
                                  +\alpha_{2k}^1\mu^{2k}+\dots,\\
  \frac{1}{\sqrt{1+\eta_2\mu^2}}&=\alpha_0^2+\alpha_2^2\mu^2+ \dots
                                  +\alpha_{2k}^2\mu^{2k}+\dots,\\
  \sqrt{\frac{1+\eta_1\mu^2}{1+\eta_2\mu^2}}&=\alpha_0^{12}+\alpha_2^{12}\mu^2+
                                              \dots
                                              +\alpha_{2k}^{12}\mu^{2k}+\dots,\\
  \sqrt{\frac{1+\eta_2\mu^2}{1+\eta_1\mu^2}}&=\alpha_0^{21}+\alpha_2^{21}\mu^2+
                                              \dots
                                              +\alpha_{2k}^{21}\mu^{2k}+\dots,
\end{align*}
where
\begin{equation*}
  \alpha_{2k}^j=\left(-\frac{\eta_j}{2}\right)^k(2k-1)!!,\qquad j=1,2,
\end{equation*}
and
\begin{equation*}
  \alpha_{2k}^{jl}=\sum_{a+b=2k}\frac{\kappa^+(\kappa^+-1)\dots(\kappa^+-a+1)}{a!}
  \cdot\frac{\kappa^-(\kappa^--1)\dots(\kappa^--b+1)}{b!}\eta_j^a\eta_l^b,
\end{equation*}
where~$\kappa^{\pm}=\pm 1/2$ and~$j=1$, $l=2$ or~$j=2$, $l=1$.

Substituting these expansions into~\eqref{eq:5} and equating coefficients at
powers of~$\mu$, we obtain the following infinite tower of $1$-dimensional
coverings:
\begin{align*}
  w_{0,x}&=\frac{iu_y}{v}\alpha_0^{21}w_0,\\
  w_{1,x}&=\frac{iu_y}{v}\alpha_0^{21}w_1+\frac{u}{2}\alpha_0^1(1+w_0^2),\\
  w_{2,x}&=\frac{iu_y}{v}(\alpha_2^{21}w_0+\alpha_0^{21}w_2) +
           u\alpha_0^1w_0w_1,\\
  w_{3,x}&=\frac{iu_y}{v}(\alpha_2^{21}w_1+\alpha_0^{21}w_3) +
           \frac{u}{2}\Big(\alpha_2^1(1+w_0^2) + \alpha_0^1(2w_0w_2+w_1^2)\Big),\\
         &\dots\\
  w_{2k,x}&=\frac{iu_y}{v}(\alpha_{2k}^{21}w_0+\alpha_{2k-2}^{21}w_2+ \dots
            \alpha_0^{21}w_{2k})\\
         &+u\Big(\alpha_{2k-2}^1w_0w_1+\alpha_{2k-4}^1(w_0w_3+w_1w_2)+\\
         & \dots
           +\alpha_0^1(w_0w_{2k-1}
           +\dots+w_{k-1}w_k)\Big),\\
  w_{2k+1,x}&=\frac{iu_y}{v}(\alpha_{2k}^{21}w_1+\alpha_{2k-2}^{21}w_3+ \dots+
              \alpha_0^{21}w_{2k+1})\\
         &+\frac{u}{2}\Big(\alpha_{2k}^1(1+w_0^2)+\alpha_{2k-2}^1(2w_0w_2+w_1^2)+\\
         &\dots +\alpha_0^1(2w_0w_{2k}+\dots+2w_{k-1}w_{k+1}+w_k^2)\Big),\\
         &\dots
           \intertext{and} w_{0,y}&=-\frac{iv_x}{u}\alpha_0^{12}w_0,\\
  w_{1,y}&=-\frac{iv_x}{u}\alpha_0^{12}w_1-\frac{iv}{2}\alpha_0^2(w_0^2-1),\\
  w_{2,y}&=-\frac{iv_x}{u}(\alpha_2^{12}w_0+\alpha_0^{12}w_2)-iv\alpha_0^2w_0w_1,\\
  w_{3,y}&=-\frac{iv_x}{u}(\alpha_2^{12}w_1+\alpha_0^{12}w_3) -
           \frac{iv}{2}\Big(\alpha_2^2(w_0^2-1)+\alpha_0^2(2w_0w_2+w_1^2)\Big),\\
  w_{2k,y}&=-\frac{iv_x}{u}(\alpha_{2k}^{12}w_0+\alpha_{2k-2}^{12}w_2+\dots
            +\alpha_0^{12}w_{2k})\\
         &-iv\Big(\alpha_{2k-2}^2w_0w_1+\alpha_{2k-4}^2(w_0w_3+w_1w_2)+\\
         & \dots +
           \alpha_0^2(w_0w_{2k-1}+\dots
           +w_{k-1}w_k)\Big),\\
  w_{2k+1,y}&=-\frac{iv_x}{u}(\alpha_{2k}^{12}w_1+\alpha_{2k-2}^{12}w_3+
              \dots\alpha_0^{12}w_{2k+1})\\
         &-\frac{iv}{2}\Big(\alpha_{2k}^2(w_0^2-1)+\alpha_{2k-2}^2(2w_0w_2+w_1^2)+\\
         &\dots +\alpha_0^2(2w_0w_{2k}+\dots+2w_{k-1}w_{k+1}+w_k^2)\Big),\\
         &\dots\\
\end{align*}

Apply the following gauge transformation in the above covering:
\begin{equation*}
  w_0=e^{\theta_0},\qquad w_k=\theta_k e^{\theta_0},\quad k>0.
\end{equation*}
Then the latter transforms to
\begin{align*}
  \theta_{0,x}&=\frac{iu_y}{v}\alpha_0^{21},\\
  \theta_{1,x}&=u\alpha_0^1\cosh\theta_0,\\
  \theta_{2,x}&=\frac{iu_y}{v}\alpha_2^{21}+u\alpha_0^1\theta_1e^{\theta_0},\\
  \theta_{3,x}&=\frac{iu_y}{v}\alpha_2^{21}\theta_1 +
                u\Big(\alpha_2^1\cosh\theta_0
                +\alpha_0^1(\theta_2+\frac{1}{2}\theta_0^2)e^{\theta_0}\Big),\\
              &\dots\\
  \theta_{2k,x}&=\frac{iu_y}{v}X_{2k}^{21}+ uX_{2k}^1,\\
  \theta_{2k+1,x}&=\frac{iu_y}{v}X_{2k+1}^{21}+ uX_{2k+1}^1,\\
              &\dots \intertext{where}
                X_{2k}^{21}&=\alpha_{2k}^{21}+\alpha_{2k-2}^{21}\theta_2+
                             \dots+\alpha_2^{21}\theta_{2k-2},\\
  X_{2k}^1&=\Big(\alpha_{2k-2}^1\theta_1 + \alpha_{2k-4}^1(\theta_3+
            \theta_1\theta_2)+ \dots+ \alpha_0^1(\theta_{2k-1} +
            \theta_1\theta_{2k-2}+\\
              & \dots+\theta_{k-1}\theta_{k})\Big)e^{\theta_0},\\
  X_{2k+1}^{21}&=\alpha_{2k}^{21}\theta_1+
                 \alpha_{2k-2}^{21}\theta_3 + \dots + \alpha_2^{21}\theta_{2k-1},\\
  X_{2k+1}^1&=\alpha_{2k}^1\cosh\theta_0+
              \Big(\alpha_{2k-2}^1(\theta_2+\frac{1}{2}\theta_1^2) +
              \alpha_{2k-4}^1(\theta_4 + \theta_1\theta_3
              +\frac{1}{2}\theta_2^2) +\\
              &\dots + \alpha_0^1(\theta_{2k}+\theta_1\theta_{2k-1} + \dots +
                \theta_{k-1}\theta_{k+1} +\frac{1}{2}\theta_k^2)\Big)e^{\theta_0},
                \intertext{and}
                \theta_{0,y}&=-\frac{iv_x}{u}\alpha_0^{12},\\
  \theta_{1,y}&=-iv\alpha_0^2\sinh\theta_0,\\
  \theta_{2,y}&=-\frac{iv_x}{u}\alpha_2^{12}-iv\alpha_0^2\theta_1e^{\theta_0},\\
  \theta_{3,y}&=-\frac{iv_x}{u}\alpha_2^{12}\theta_1 -
                iv\Big(\alpha_2^2\sinh\theta_0 +
                \alpha_0^2(\theta_2+\frac{1}{2}\theta_1^2)e^{\theta_0}\Big),\\
              &\dots\\
  \theta_{2k,y}&=-\frac{iv_x}{u}Y_{2k}^{12}-ivY_{2k}^2,\\
  \theta_{2k+1,y}&=-\frac{iv_x}{u}Y_{2k+1}^{12}-ivY_{2k+1}^2,
                   \intertext{where} Y_{2k}^{12}
              &=\alpha_{2k}^{12}+\alpha_{2k-2}^{12}\theta_2 +
                \dots +
                \alpha_2^{12}\theta_{2k-2},\\
  Y_{2k}^2&=\Big(\alpha_{2k-2}^2\theta_1 + \alpha_{2k-4}(\theta_3 +
            \theta_1\theta_2) +\\
              & \dots+ \alpha_0^2(\theta_{2k-1} +
                \theta_1\theta_{2k-2}
                +\dots+ \theta_{k-1}\theta_k)\Big)e^{\theta_0},\\
  Y_{2k+1}^{12}&=\alpha_{2k}^{12}\theta_1+\alpha_{2k-2}^{12}\theta_3 +\dots+
                 \alpha_2^{12}\theta_{2k-1},\\
  Y_{2k+1}^2&=\alpha_{2k}^2\sinh\theta_0 + \Big(\alpha_{2k-2}^2(\theta_2 +
              \frac{1}{2}\theta_1^2) + (\theta_4 + \theta_1\theta_3
              +\frac{1}{2}\theta_2^2)+\\
              &\dots + \alpha_0^2(\theta_{2k}+\theta_1\theta_{2k-1} +\dots+
                \theta_{k-1}\theta_{k+1} +\frac{1}{2}\theta_k^2)\Big)e^{\theta_0}.
\end{align*}

\begin{remark}
  The above complex covering can be reduced to the real form if we
  set~$\theta_k=p_k + iq_k$, $k=0,1,\dots$, and consider real and imaginary
  parts of the defining equations separately. In the particular case~$p_0=0$
  we obtain a real covering which is employed below (see Section~\ref{sec-ro})
  to construct a recursion operator for symmetries of \eqref{sys-eta}.
\end{remark}

Thus, we have obtained an infinite tower of Abelian coverings
\begin{equation*}
  \xymatrix{
    \mathcal{E}&\ar[l]\mathcal{E}_0&\ar[l]\dots&\ar[l]\mathcal{E}_{2k}
    &\ar[l]\mathcal{E}_{2k+1}&\ar[l]\dots,
  }
\end{equation*}
where~$\mathcal{E}$ is the initial system (\ref{sys-eta}) and~$\mathcal{E}_s$
is obtained by extending~$\mathcal{E}$ with the nonlocal
variables~$\theta_0,\dots,\theta_s$. It only remains to prove that for any
$s\in\mathbb{N}$ the conservation law
\begin{equation}\label{eq:9}
  \omega_s=\Big(\frac{iu_y}{v}X_s^{21} + uX_s^1\Big)\d x -
  \Big(\frac{iv_x}{u}Y_s^{12} + ivY_s^2\Big)\d y
\end{equation}
is nontrivial on~$\mathcal{E}_{s-1}$.

Denote by~$D_x^{[l]}$, $D_y^{[l]}$ the total derivatives on~$\mathcal{E}_l$,
$l\geq 1$.

\begin{proposition}
  \label{prop1}
  The only solutions of the system
  \begin{equation}\label{eq:6}
    D_x^{[l]}(f)=0,\quad D_y^{[l]}(f)=0,\qquad f\in C^\infty(\mathcal{E}_l),
  \end{equation}
  are constants.
\end{proposition}

\begin{proof}[Proof of Proposition~\ref{prop1} (Part I)]
  The total derivatives on~$\mathcal{E}_l$ have the form
  \begin{equation*}
    D_x^{[l]}=D_x+\sum_{j=0}^l\Big(\frac{iu_y}{v}X_s^{21} +
    uX_s^1\Big)\frac{\partial}{\partial\theta_s}, \quad
    D_y^{[l]}=D_y- \sum_{j=0}^l\Big(\frac{iv_x}{u}Y_s^{12} +
    ivY_s^2\Big)\frac{\partial}{\partial\theta_s},
  \end{equation*}
  where~$D_x$ and~$D_y$ are the total derivatives on~$\mathcal{E}$. Obviously,
  a function~$f\in C^\infty(\mathcal{E}_l)$ is a solution to~\eqref{eq:6} if
  and only if
  \begin{equation*}
    \frac{\partial f}{\partial x}=0,\quad
    \frac{\partial f}{\partial y}=0,\quad
    \frac{\partial f}{\partial u}=0,
  \end{equation*}
  and
  \begin{equation}\label{eq:8}
    \sum_{s=0}^lX_s^{12}\frac{\partial f}{\partial\theta_s}=0,\quad
    \sum_{s=0}^lY_s^{21}\frac{\partial f}{\partial\theta_s}=0,\quad
    \sum_{s=1}^lX_s^1\frac{\partial f}{\partial\theta_s}=0,\quad
    \sum_{s=1}^lY_s^2\frac{\partial f}{\partial\theta_s}=0.
  \end{equation}
  We shall proceed with the proof after establishing Lemma~\ref{lemm1}.
\end{proof}

Consider the vector fields
\begin{equation*}
  TX_l=\sum_{s=1}^lX_s^1\frac{\partial}{\partial\theta_s},\qquad
  TY_l=\sum_{s=1}^lY_s^2\frac{\partial}{\partial\theta_s}
\end{equation*}
on~$\mathcal{E}_l$ and let
\begin{equation*}
  Z_l=[[\dots[TX_l,\underbrace{TY_l],\dots],TY_l]}_{l-1\ \text{times}}
\end{equation*}
for~$l\geq 2$.

\begin{lemma}
  \label{lemm1}
  For any~$l\geq 2$ we have
  \begin{equation*}
    Z_l=
    \begin{cases}
      \alpha_0^1(\alpha_0^2)^{l-1}\displaystyle\frac{\partial}{\partial\theta_l}&\text{for even\ }l,\\
      \alpha_0^1(\alpha_0^2)^{l-1}\cosh(\theta_0)\displaystyle\frac{\partial}{\partial\theta_l}&\text{for odd\ }l.
    \end{cases}
  \end{equation*}
\end{lemma}

\begin{proof}[Proof of Lemma~\ref{lemm1}]
  Consider the formal series
  \begin{equation*}
    Q=\frac{1}{2}\left(1+\sum_{j=1}^\infty\theta_j\lambda^j\right)^2
  \end{equation*}
  and denote by~$\s{j}$ its coefficient at~$\lambda^j$.

  Assign to any monomial~$p=\phi(\theta_0)\theta_{i_1}\dots\theta_{i_r}$ the
  weight
  \begin{equation*}
    \abs{p}=i_1+\dots+i_r.
  \end{equation*}
  Then the quantities~$\s{j}$ are homogeneous and~$\abs{\s{j}}=j$. In
  addition, one has
  \begin{equation}\label{eq:7}
    \frac{\partial\s{j}}{\partial\theta_i}=
    \begin{cases}
      1,&\text{if\ }i=j,\\
      \theta_{j-i}&\text{if\ }i<j,\\
      0,&\text{otherwise}
    \end{cases}
  \end{equation}
  for any~$i\geq 1$.

  Consider the vector fields
  \begin{equation*}
    TX=\sum_{s=1}^\infty X_s^1\frac{\partial}{\partial\theta_s},\qquad
    TY=\sum_{s=1}^\infty Y_s^2\frac{\partial}{\partial\theta_s}
  \end{equation*}
  on the space~$\mathcal{E}_\infty=\liminv_{l\to\infty}\mathcal{E}_l$. Then
  obviously,~$\eval{TX}_{\mathcal{E}_l}=TX_l$,
  $\eval{TY}_{\mathcal{E}_l}=TY_l$
  and~$\eval{[TX,TY]}_{\mathcal{E}_l}=[TX_l,TY_l]$. These fields can be
  written as follows:
  \begin{align*}
    TX&=\alpha_0^1\left(\cosh(\theta_0)\frac{\partial}{\partial\theta_1} +
        e^{\theta_0}\left((\s{1}+o(0))\frac{\partial}{\partial\theta_2} +
        (\s{2}+o(1))\frac{\partial}{\partial\theta_3}+\dots\right.\right.\\
      &\left.\left.\dots+(\s{l-1}+o(l-2))\frac{\partial}{\partial\theta_l}
        +\dots\right)\right),\\
    TY&=\alpha_0^2\left(\sinh(\theta_0)\frac{\partial}{\partial\theta_1} +
        e^{\theta_0}\left((\s{1}+o(0))\frac{\partial}{\partial\theta_2} +
        (\s{2}+o(1))\frac{\partial}{\partial\theta_3}+\dots\right.\right.\\
      &\left.\left.\dots+(\s{l-1}+o(l-2))\frac{\partial}{\partial\theta_l}
        +\dots\right)\right),
  \end{align*}
  where~$o(j)$ denotes terms of weight~$\leq j$. To simplify notation, we
  shall use the short form
  \begin{align*}
    TX&=\alpha_0^1\left(\cosh(\theta_0)\frac{\partial}{\partial\theta_1} +
        e^{\theta_0}\left(\s{1}\frac{\partial}{\partial\theta_2} +
        \s{2}\frac{\partial}{\partial\theta_3}
        +\dots+\s{l-1}\frac{\partial}{\partial\theta_l}
        +\dots\right)\right)+o,\\
    TY&=\alpha_0^2\left(\sinh(\theta_0)\frac{\partial}{\partial\theta_1} +
        e^{\theta_0}\left(\s{1}\frac{\partial}{\partial\theta_2} +
        \s{2}\frac{\partial}{\partial\theta_3}
        +\dots+\s{l-1})\frac{\partial}{\partial\theta_l}
        +\dots\right)\right)+o.
  \end{align*}
  One readily checks that
  \begin{equation*}
    [A+o,B+o]=[A,B]+o
  \end{equation*}
  in all subsequent computations.

  We shall now prove, by induction on~$l$, that the fields
  \begin{equation*}
    Z_l=[[\dots[TX,\underbrace{TY],\dots],TY]}_{l-1\ \text{times}}
  \end{equation*}
  are of the form
  \begin{equation*}
    \alpha_0^1(\alpha_0^2)^{l-1}\left(\frac{\partial}{\partial\theta_l} +
      \theta_1\frac{\partial}{\partial{\theta_{l+1}}} + \dots +
      \theta_j\frac{\partial }{\partial\theta_{l+j}} + \dots\right) +o
  \end{equation*}
  if~$l$ is even and
  \begin{equation*}
    \alpha_0^1(\alpha_0^2)^{l-1}\left(\cosh(\theta_0)
      \frac{\partial}{\partial\theta_l} +
      e^{\theta_0}\left(\s{1}\frac{\partial}{\partial \theta_{l+1}} + \dots +
        \s{j}\frac{\partial}{\partial\theta_{l+j}}+\dots\right)\right) +o
  \end{equation*}
  for odd~$l$. The claim of Lemma~\ref{lemm1} readily
follows from these formulas for $Z_l$ (please note, however, that the $Z_l$ in the statement
of Lemma~\ref{lemm1} are not identical to those defined above).

  Let~$l=2$. Then by virtue of \eqref{eq:7} we have
  \begin{multline*}
    Z_2=[TX,TY]\\
    =\left[\cosh(\theta_0)\frac{\partial}{\partial\theta_1} + e^{\theta_0}
      \sum_{j=1}^\infty \s{j+1}\frac{\partial}{\partial\theta_j},
      \sinh(\theta_0)\frac{\partial}{\partial\theta_1} + e^{\theta_0}
      \sum_{j=1}^\infty \s{j+1}\frac{\partial}{\partial\theta_j}\right]+o \\
    =\alpha_0^1\alpha_0^2e^{\theta_0}\left(\cosh(\theta_0)-\sinh(\theta_0)\right)
    \left(\frac{\partial}{\partial\theta_{2}}
      +\sum_{j=1}^\infty\theta_j\frac{\partial}{\partial\theta_{j+2}}\right)+o\\
    =\alpha_0^1\alpha_0^2\left(\frac{\partial}{\partial\theta_2}+
      \sum_{j=1}^\infty
      \theta_j\frac{\partial}{\partial\theta_{j+2}}\right)+o.
  \end{multline*}
  For~$l=3$ one has
  \begin{multline*}
    Z_3=[Z_2,TY]\\
    =\left[\alpha_0^1\alpha_0^2\left(\frac{\partial}{\partial\theta_2}+
        \sum_{j=1}^\infty
        \theta_j\frac{\partial}{\partial\theta_{j+2}}\right), \alpha_0^2\left(
        \sinh(\theta_0)\frac{\partial}{\partial\theta_1} + e^{\theta_0}
        \sum_{j=1}^\infty
        \s{j+1}\frac{\partial}{\partial\theta_j}\right)\right] +o\\=
    \alpha_0^1(\alpha_0^2)^2e^{\theta_0}\left(\frac{\partial}{\partial\theta_3}
      + \sum_{j=1}^\infty\theta_j\frac{\partial}{\partial\theta_{j+3}} +
      \sum_{s=1}^\infty\theta_s\left(\frac{\partial}{\partial\theta_{s+3}} +
        \sum_{j=1}^\infty\theta_j\frac{\partial}{\partial\theta_{j+s+3}}\right)
    \right)\\
    -\alpha_0^1(\alpha_0^2)^2
    \left(\sinh(\theta_0)\frac{\partial}{\partial\theta_3} + e^{\theta_0}
      \sum_{j=1}^\infty \s{j}\frac{\partial}{\partial\theta_{j+3}}\right)
    + o \\
    =\alpha_0^1(\alpha_0^2)^2e^{\theta_0}\left(\frac{\partial}{\partial\theta_3}
      + 2\sum_{j=1}^\infty \s{j}\frac{\partial}{\partial\theta_{j+3}}
    \right)\\
    -\alpha_0^1(\alpha_0^2)^2
    \left(\sinh(\theta_0)\frac{\partial}{\partial\theta_3} + e^{\theta_0}
      \sum_{j=1}^\infty \s{j}\frac{\partial}{\partial\theta_{j+3}}\right)
    + o\\
    =\alpha_0^1(\alpha_0^2)^2\left(\cosh(\theta_0)
      \frac{\partial}{\partial\theta_3} +
      e^{\theta_0}\left(\s{j}\frac{\partial}{\partial\theta_{j+3}}\right)\right)
    + o.
  \end{multline*}

  The induction step also uses property~\eqref{eq:7} of the quantities~$\s{j}$
  and is accomplished by the computations quite similar to those given above.
\end{proof}

Let us complete the proof of Proposition~\ref{prop1}.

\begin{proof}[Proof of Proposition~\ref{prop1} (Part II)]
  We prove by induction on~$l$ that system~\eqref{eq:6} possesses constant
  solutions only.

  For~$l=1$, the equations are
  \begin{align*}
    &D_x(f)+\alpha_0^{21}\frac{iu_y}{v}\frac{\partial f}{\partial\theta_0} +
      \alpha_0^1\cosh(\theta_0)u\frac{\partial f}{\partial\theta_1}=0,\\
    &D_y(f)-\alpha_0^{12}\frac{iv_x}{u}\frac{\partial f}{\partial\theta_0} -
      \alpha_0^2i\sinh(\theta_0)v\frac{\partial f}{\partial\theta_1}=0,
  \end{align*}
  where~$f$ is a function on~$\mathcal{E}_1$, i.e., it may depend on jet
  variables as well as on~$\theta_0$ and~$\theta_1$. Analyzing the
  coefficients of the fields~$D_x$ and~$D_y$, one immediately sees that~$f$
  can depend on~$x$, $y$, $\theta_0$, and~$\theta_1$ only and from the above
  equations it readily follows that~$f=\const$.

  Assume now that the result is valid for some~$l>1$ and
  let~$f\in C^\infty(\mathcal{E}_{l+1})$ be such
  that~$D_x^{[l+1]}(f)=D_y^{[l+1]}(f)=0$. Then, by~\eqref{eq:8},
  \begin{equation*}
    TX_{l+1}(f)=0,\qquad TY_{l+1}(f)=0.
  \end{equation*}
  Consequently,~$Z_{l+1}(f)=0$. Using Lemma~\ref{lemm1}, we see that
  \begin{equation*}
    \frac{\partial f}{\partial \theta_{l+1}}=0.
  \end{equation*}
  Thus,~$f$ is a function on~$\mathcal{E}_l$ and it is constant by the
  induction hypothesis.
\end{proof}

\begin{corollary}
  The conservation law~$\omega_s$ given by~\eqref{eq:9} is nontrivial
  on~$\mathcal{E}_s$.
\end{corollary}

\begin{proof}
  Indeed, if~$\omega_s=\d_h f$,
  where~$\d_h=\d x\wedge D_x^{[s]}+\d y\wedge D_y^{[s]}$ is the horizontal
  de~Rham differential on~$\mathcal{E}_s$ then~$D_x^{[s+1]}(f-\theta_{s+1})=0$
  and~$D_x^{[s+1]}(f-\theta_{s+1})=0$, which contradicts to
  Proposition~\ref{prop1}.
\end{proof}

Thus we have constructed an infinite hierarchy of nonlocal conservation laws
$\omega_s$, $s\in\mathbb{N}$, for (\ref{sys-eta}).

\section{Recursion operator and Hamiltonian structures}\label{sec-ro}
It is readily checked that~\eqref{sys-eta} possesses an Abelian covering
$\mathcal{S}$ with the fiber coordinates $s_i$ defined by the formulas
\begin{align*}
  (s_0)_x&=\frac{u_y}{v},
  &&(s_0)_y=-\frac{v_x}{u}, \\ (s_1)_x&=\cos(s_0)u,
  &&(s_1)_y=v \sin(s_0), \\ (s_2)_x&=-\sin(s_0)u, &&(s_2)_y=v \cos(s_0).
\end{align*}
Note that the conservation laws associated with $s_i$, $i=1,2$, are potential
conservation laws in the terminology of \cite{kp} and that we have
$\theta_0=i s_0$ and $\theta_1=\alpha_0^1 s_1$ modulo the addition of
arbitrary constants.

The following result is readily verified by straightforward computation.

\begin{proposition}\label{pr-ro}
  Suppose that $U$,$V$ and $S_i$ are fiber coordinates of the tangent covering
  $V\mathcal{S}$. 

  Then the tangent covering over (\ref{sys-eta}) admits a B\"acklund
  auto-transformation (i.e., a recursion operator for (\ref{sys-eta})) of the
  form
  \begin{equation}\label{ro}
    \begin{array}{ll}
      \tilde{U}&=\eta_1 U+\left(\ds\frac12\eta'_1\cos(s_0)
                 +\frac{\sin(s_0)(\eta_2-\eta_1)}{v}u_y\right)S_1\\
               &+\left(\ds-\frac12\eta'_1\sin(s_0)+
                 \frac{\cos(s_0)(\eta_2-\eta_1)}{v}u_y\right)S_2,\\[5mm]
      \tilde{V}&=\eta_2
                 V+\left(\ds\frac12\eta'_2\sin(s_0)-
                 \frac{\cos(s_0)(\eta_2-\eta_1)}{u}v_x\right)S_1\\
               &+\left(\ds\frac12\eta'_2\cos(s_0)+
                 \frac{\sin(s_0)(\eta_2-\eta_1)}{u}v_x\right)S_2.
    \end{array}
  \end{equation}
\end{proposition}

Equations (\ref{ro}) define a recursion operator for (\ref{sys}) in the
following fashion (see e.g.\ \cite{marvan, ms, ms2, asro} and references
therein for details).

Suppose $(U,V)^T$ is a symmetry shadow for (\ref{sys}) in a covering
$\mathcal{C}$ over $\mathcal{S}$ (here and below the superscript $T$
indicates the transposed matrix). Then we have a (possibly trivial)
covering $\mathcal{C}'$ over $\mathcal{C}$ arising from substituting
our $U$ and $V$ into the equations defining $S_i$. Under these
assumptions (\ref{ro}) defines a new symmetry shadow
$(\tilde{U},\tilde{V})^T$ for (\ref{sys}) in $\mathcal{C}'$, i.e.,
we have a recursion operator for (\ref{sys}).

Note that (\ref{ro}) was found using the method from \cite{ms}; cf.\ e.g.\
\cite{marvan, kvv, m2, mo, asro} for some other related techniques.

Starting with a simple seed symmetry like $(0,0)^T$ yields, through
the repeated application of (\ref{ro}), an infinite hierarchy of
shadows of nonlocal symmetries for (\ref{sys-eta}).  It is an
interesting open problem to find the minimal covering in which it is
possible to lift all these shadows to full-fledged nonlocal
symmetries of (\ref{sys-eta}) and to find the commutation relations
among these nonlocal symmetries, cf.\ e.g.\ \cite{mos}.

Moreover, one can readily establish the following
result:
\begin{proposition}
  \label{prop2}
  System (\ref{sys-eta}) admits (in a generalized sense of \cite{kkv,
    kvv,k-v}) a pair of compatible local Hamiltonian structures
  $\mathcal{P}_i$ of the form
  \[
  \mathcal{P}_i=f_{i,3} D_x D_y + D_y\circ f_{i,2} +  D_x\circ
  f_{i,1}+f_{i,0},\quad i=1,2,
  \]
  where $f_{i,j}$ are $2\times 2$ matrices of the form
  \[
  f_{1,3}=\left(\begin{array}{cc}1
                  & 0\\0
                  & -1\end{array}\right),\quad f_{2,3}
              =\left(\begin{array}{cc}\eta_1
                       & 0\\0
                       & -\eta_2\end{array}\right)
   \]
   \[
   f_{1,2}=\left(\begin{array}{cc}0
                   & 0\\[5mm]
                   \displaystyle\frac{v_x}{u}+\frac{v\eta'_1}{2 u\Delta}
                  &\displaystyle-\frac{u_x}{u}-\frac{\eta'_1}{2
                    \Delta} \end{array}\right),\]

  \[
  f_{2,2}=\left(\begin{array}{cc}\displaystyle\frac{\eta'_1}{2}
                  & 0\\[5mm]
                  \displaystyle\frac{\eta_1 v_x}{u}+\frac{v\eta'_1\eta_2}{2
                  u\Delta}

                  & \displaystyle-\frac{\eta_2
                             u_x}{u}-\frac{\eta'_1\eta_2}{2
                             \Delta} \end{array}\right),
   \]

   \[
   f_{1,1}=\left(\begin{array}{cc}\displaystyle
                   \frac{v_y}{v}-\frac{\eta'_2}{2\Delta}
                   &
                     \displaystyle-\frac{u_y}{v}+\frac{\eta'_2
                     u}{2v\Delta}\\[5mm] 0
                   & 0\end{array}\right),\]

  \[
  f_{2,1}=\left(\begin{array}{cc}\displaystyle \frac{\eta_1 v_y}{v}-\frac{\eta_1
                  \eta'_2}{2\Delta}
                  & \displaystyle-\frac{\eta_2 u_y}{v}+\frac{\eta_1
                    \eta'_2 u}{2v\Delta}\\[5mm] 0
                  &
                    \displaystyle-\frac{\eta'_2}{2}\end{array}\right)
  \]
  \[
  f_{1,0}=\left(\begin{array}{cc}\displaystyle\frac{\eta'_1 u_y}{2 u \Delta}
                  &
                    \displaystyle \frac{u_{xy}}{v}-\frac{u_x u_y}{v u}-\frac{v_x
                    u_y}{v^2}-\frac{\eta'_1 u_y}{2 v\Delta}\\[5mm]
                  \displaystyle-\frac{v_{xy}}{u}+\frac{v_x v_y}{v u}+\frac{v_x
                  u_y}{u^2}-\frac{\eta'_2 v_x}{2 u \Delta}
                  &\displaystyle\frac{\eta'_2
                  v_x}{2 v\Delta}
                \end{array}\right)
  \]
  \medskip
  \[
  f_{2,0}=\left(\begin{array}{lr}\displaystyle-\frac{\eta'_1 v_y}{2 v}
                  +\frac{\eta'_1 \eta_2 u_y}{2 u\Delta}
                +\frac{\eta'_1\eta'_2}{4\Delta}
                  & (f_{2,0})_{12}\\[5mm](f_{2,0})_{21}
                  &\displaystyle
                    \frac{\eta'_2 u_x}{2 u}+\frac{\eta_1 \eta'_2 v_x}{2 v\Delta}
                    +\frac{\eta'_1\eta'_2}{4\Delta}
              \end{array}\right),
  \]
  \medskip
  \[
  (f_{2,0})_{12}=\displaystyle \frac{\eta_2 u_{xy}}{v}
  -\frac{\eta_2 u_x u_y}{v u}+\frac{\eta_2 v_x u_y}{v^2}-\frac{\eta'_1 \eta_2
    u_y}{2 v \Delta}-\frac{\eta'_1\eta'_2 u}{4v\Delta},
  \]
  \[
  (f_{2,0})_{21}=\displaystyle-\frac{\eta_1 v_{xy}}{u}
  +\frac{\eta_1 v_x v_y}{v u}+\frac{\eta_1 v_x u_y}{u^2}-\frac{\eta_1 \eta'_2
    v_x}{2 u \Delta}-\frac{\eta'_1\eta'_2 v}{4u\Delta},
  \]
  and $\Delta=\eta_2-\eta_1$.
\end{proposition}

\begin{proof}
Our system (\ref{sys-eta}) is readily shown to be contact equivalent
to an evolutionary one (cf.\ Section~\ref{sec-ho} below).
Consequently, system (\ref{sys-eta}) is normal and its cotangent
covering (the $\ell^*$-covering) is well-defined, cf.\ \cite{kvv},
and the Schouten brackets of variational bivectors can be computed
in any representation of the system under study. Let us employ to
this end an evolutionary representation.\looseness=-1

Namely, let $\mathcal{E}$ denote now an evolutionary representation
for our system. Let $u_\sigma^j$ be the jet variables on
$\mathcal{E}$ and $p_\sigma^j$ be the associated variables in the
fiber of the cotangent covering of $\mathcal{E}$. Then to any
differential operator $\Delta=\left\Vert\sum_\sigma
a_\sigma^{ij}D_\sigma\right\Vert$ in total derivatives acting from
the space of symmetries of $\mathcal{E}$ to that of cosymmetries of
$\mathcal{E}$ there corresponds the function
  \begin{equation}\label{w-delta}
    W_\Delta=\sum_{i,j,\sigma}a_\sigma^{ij}p_\sigma^j p^i
  \end{equation}
on the space of the cotangent covering, see~\cite{l-star}.

The operator $\Delta$ is skew-adjoint (and to the function
$W_\Delta$ we can associate a variational bivector, cf.\
\cite{l-star}) if and only if \cite{l-star}
  \begin{equation}
    \label{eq:10}
    \sum_j\frac{\delta W_\Delta}{\delta p^j}p^j=-2W_\Delta,
  \end{equation}
while a variational bivector is a Poisson structure (and the
corresponding operator is Hamiltonian) if and only if \cite{l-star}
\begin{equation}
    \label{eq:11}
    \delta\sum\frac{\delta W}{\delta u^j}\frac{\delta W}{\delta p^j}=0,
  \end{equation}
where~$\delta$ is the Euler operator.

Finally, two Poisson structures are compatible if and only if the
associated bivectors satisfy \cite{l-star}
  \begin{equation}
    \label{eq:12}
    \delta\sum_j\left(\frac{\delta W_1}{\delta u^j}\frac{\delta W_2}{\delta
        p^j} + \frac{\delta W_2}{\delta u^j}\frac{\delta W_1}{\delta
        p^j}\right)=0.
  \end{equation}

Identities~\eqref{eq:10}--\eqref{eq:12}, although quite complicated,
are easily checked for the quantities $W_i$ associated to the
counterparts of $\mathcal{P}_i$ for $\mathcal{E}$ by an appropriate
computer algebra software.
\end{proof}

The ratio of the Hamiltonian structures
$\mathcal{R}=\mathcal{P}_2\circ \mathcal{P}_1^{-1}$ is nothing but
the recursion operator (\ref{ro}) written in the pseudodifferential
form, cf.\ e.g.\ \cite{blaszak, marvan, m, ms, asbih} and references
therein. The compatibility of $\mathcal{P}_1$ and $\mathcal{P}_2$
implies that $\mathcal{R}$ is hereditary, cf.\ e.g.\ \cite{blaszak,
k} and references therein for details.

\section{Compatible Hamiltonian structures in evolutionary
  representation}\label{sec-ho}
In order to make contact with the standard theory of Hamiltonian structures,
see e.g.\ \cite{blaszak} and references therein, we should rewrite
(\ref{sys-eta}) in evolutionary form and find the counterparts of
$\mathcal{P}_i$ for this new form.

First of all, in order to simplify further computations note \cite{evfpc} that
in general position we can without loss of generality assume that
$\eta_i=x^i$.

To see this, introduce new independent variables $\tilde x^i=\eta_i(x^i)$, and
let $h_i$ stand for the inverse functions, i.e., $x^i=h_i(\tilde x^i)$, and
new dependent variables $\tilde H_i$ such that $H_i=\eta'_i \tilde H_i$ (no
sum over $i$). Next, for $i\neq j$ put
$\tilde\beta_{ij}=\tilde\partial_i\tilde H_j/\tilde H_i$, where
$\tilde\partial_i=\partial/\partial\tilde x^i$.

We find that $\tilde H_i=\tilde H_i(\tilde x^1,\tilde x^2)$ satisfy
\begin{equation*}
  \begin{array}{c}
    \tilde\partial_1\tilde H_2=\tilde\beta_{12}\tilde H_1,
    ~~~ \tilde\partial_2\tilde H_1=\tilde\beta_{21}\tilde H_2, \\
    \ \\
    \tilde\partial_1\tilde\beta_{12}+\tilde\partial_2\tilde\beta_{21}=0, \\
    \ \\
    \tilde x^1 \tilde\partial_1\tilde\beta_{12}+\tilde x^2
    \tilde\partial_2\tilde\beta_{21}+\frac{1}{2}\tilde\beta_{12}
    +\frac{1}{2}\tilde\beta_{21}+\tilde H_1 \tilde H_2=0.
  \end{array}
  \label{fsa}
\end{equation*}

The above system is, modulo the tildes, nothing but a special case of
(\ref{eq:2}) with $\eta_i=x^i$. Thus, in what follows we shall assume without
loss of generality that $\eta_1=x^1$ and $\eta_2=x^2$.

With this assumption in mind, we shall work below with the system
\begin{equation}\label{sys}
  \begin{array}{rcl}
    u_{yy} &=& \displaystyle\frac{\frac{1}{2}v^2 v_x+\frac{1}{2}u v
               u_y +(x-y) u u_y v_y+u^2 v^3} {u v (x-y)},\\
    v_{xx} &=& \displaystyle\frac{\frac{1}{2}u v v_x +\frac{1}{2}u^2 u_y+ (y-x)
               v v_x u_x+u^3 v^2}{u v(y-x)}.
  \end{array}
\end{equation}
instead of (\ref{sys-eta}).

In order to study the Hamiltonian structures admitted by (\ref{sys}), perform
the following change of variables: put $z=(x+y)/2$ and
$t=(x-y)/2$. Then~\eqref{sys} can be rewritten in an evolutionary form:
\begin{equation}\label{ev}
  \begin{array}{ll}
    u_t&=p,\\
    v_t&=q,\\
    p_t&=
         \ds-u_{zz}+ 2 p_z + \frac{u_z v_z}{v} - \frac{(v + 2 z q) u_z}{2 v z} -
         \frac{(2 z u p - v^2) v_z}{2 u v z} \\[2mm]
       &+ \dfrac{2 u^2 v^3 + u v p + 2 z u p q +  v^2 q}{2 u v z},\\[5mm]
    q_t&=
         \ds-v_{zz} - 2 q_z+ \frac{u_z v_z}{u} + \frac{(u^2 + 2 z v q) u_z}{2
         u v z} -
         \frac{(u - 2 z p) v_z}{2 z u}  \\[2mm]
       &- \dfrac{2 u^3 v^2 + u^2 p + u v q - 2  z v p q}{2 u v z}.
  \end{array}
\end{equation}
Introduce $4\times 4$ matrices $g_i$ of the form
\[
\ds
g_3=\left(\begin{array}{@{}cccc@{}} 0
            & 0 & 0 & 0 \\ 0
            & 0 & 0 & 0 \\ 0 &
                               0 &
                                   2(z + t) & 0 \\
            0 & 0 & 0 & 2(t-z) \end{array}\right),\]
\medskip
\[
g_2=\left(\begin{array}{@{}cccc@{}} 0 & 0 & 2(z + t) & 0
            \\ 0 & 0 & 0 &
                           2(z - t) \\ -2(z+t) & 0 & \ds 3 &
                                                             (g_2)_{34}\\
            0 & 2(t-z) &
                         -(g_2)_{34} &
                                       -3 \end{array}\right),
\]
where
\[
(g_2)_{34}=-\frac{(z - t) u_z}{v} - \frac{(z + t) v_z}{u} - \frac{(z + t)
  q}{u} + \frac{(z - t) p}{v},
\]
\[
\ds
g_1=\left(\begin{array}{@{}cccc@{}} -2(z + t)
            & 0 & 1 &
                      (g_1)_{14}\\
            0 & 2(z - t)
                &
                  -(g_1)_{14}&
                               1 \\ -3 &
                                         -(g_1)_{14} &
                                                       (g_1)_{33} &
                                                                    (g_1)_{34}\\
            (g_1)_{14}&
                        -3&
                            (g_1)_{43} &
                                         (g_1)_{44}
          \end{array}\right)
\]
where
\begin{align*}
  (g_1)_{14}&=\ds\frac{(z - t) u_z}{v} - \frac{(z + t) v_z}{u} - \frac{(z + t)
              q}{u} - \frac{(z - t) p}{v},\\[5mm]
  (g_1)_{33}&=\ds2(z + t) \left(-\frac{u_z v_z}{u v} - \frac{q u_z}{u v} +
              \frac{p v_z}{u v} + 2 \frac{v^2}{z} + \frac{p q}{u
              v}\right),\\[5mm]
  (g_1)_{34}&=\ds\frac{(z- t) u_{zz}}{v} - 3 \frac{(z + t) v_{zz}}{u} +
              \frac{(z (3v^2-u^2) + t (u^2 + 3 v^2)) u_z v_z}{u^2 v^2}\\[5mm]
            &\ds+ 3 \frac{(z + t) (q u_z-u q_z)}{u^2} - \frac{((t-z) u p + 2
              v^2) v_z}{v^2 u} + \frac{(t-z) p_z}{v} - 2 \frac{u^2 v +
              q}{u},\\[5mm]
  (g_1)_{43}&=\ds3 \frac{(z - t) u_{zz}}{v} - \frac{(z + t) v_{zz}}{u} +
              \frac{(3 (t-z) u^2+ (t+z) v^2) u_z v_z}{u^2 v^2}\\[5mm]
            &\ds  + \frac{(2 u^2 + (z+t) v q) u_z}{u^2 v}+3 \frac{(z - t) (p
              v_z-v p_z)}{v^2} - \frac{(z + t) q_z}{u} - 2 \frac{u v^2 +
              p}{v}, \\[5mm]
  (g_1)_{44}&=\ds\frac{(t-z)}{(t+z)}(g_1)_{33}+4 \frac{(u^2 + v^2)(z-t)}{z},
\end{align*}
and
\[
g_0=\left(\begin{array}{@{}cccc@{}} -1
            & (g_0)_{12} & (g_0)_{13} &(g_0)_{14}\\
            -(g_0)_{12}
            & 1 &(g_0)_{23}&(g_0)_{24}\\
            -(g_0)_{13} &  (g_0)_{32} &(g_0)_{33} &(g_0)_{34}\\
            (g_0)_{41}& -(g_0)_{24} &(g_0)_{43} &(g_0)_{44}
          \end{array}\right),
\]
where
\begin{align*}
  (g_0)_{12} &= \ds\frac{(z - t) u_z}{v} + \frac{(z + t) v_z}{u} + \frac{(t-z)
               u p + (t+ z) v q}{u v},\\[5mm]
  (g_0)_{13} &= \ds-\frac{(z - t) u_z^2}{v^2} - \frac{(z + t) u_z v_z}{u v} -
               \frac{(2 (t-z) u p + (z+t) v q) u_z}{u v^2} \\[5mm]
             &\ds  + \frac{(z + t) p v_z}{u v}+\frac{(z+t) u v^4+ z (t-z) u
               p^2 + z (z+t) v p q}{v^2 u z},\\[5mm]
  (g_0)_{14}&=\ds-2 \frac{(z + t) v_{zz}}{u} + 2 \frac{(z + t) u_z v_z}{u^2} +
  2 \frac{(z + t) q u_z}{u^2} - \frac{v_z}{u}\\[5mm] & - 2 \frac{(z + t)
                                                       q_z}{u} -
  \frac{z u^2 v + t u^2 v + z q}{z u},
  \\[5mm]
  (g_0)_{23}&=\ds-2 \frac{(z - t) u_{zz}}{v} + 2 \frac{(z - t) u_z v_z}{v^2} -
  \frac{u_z}{v} - 2 \frac{(z - t) p v_z}{v^2}\\[5mm] & + 2 \frac{(z - t)
                                                       p_z}{v} +
  \frac{(z- t) u v^2 + z p}{v z},
  \\[5mm]
  (g_0)_{24}&=\ds-\frac{(z - t) u_z v_z}{u v} - \frac{(z + t) v_z^2}{u^2} -
  \frac{(z - t) q u_z}{u v} - \frac{((t-z) u p + 2 (z+t) v q) v_z}{u^2 v}
  \\[5mm]
  &-\ds \frac{(z-t) u^4 v + z (t-z) u p q + z(z+t) v q^2}{z u^2 v},
  \\[5mm]
  (g_0)_{32} &= \ds\frac{(z - t) u_{zz}}{v} + \frac{(z + t) v_{zz}}{u} -
  \frac{(z u^2 - t u^2 + z v^2 + t v^2) u_z v_z}{u^2 v^2} - \frac{(z + t) q
    u_z}{u^2}\\[5mm]
  & + \frac{((z-t) u p + v^2) v_z}{u v^2} -\frac{(z - t) p_z}{v}
  + \frac{(z + t) q_z}{u} +\frac{(t-z) u^2 v + z
    q}{z u},
  \\[5mm]
  (g_0)_{33} &=\ds-\frac{(z + t) v_z u_{zz}}{u v} - \frac{(z + t) (q+v_z) u_{zz}}{u
    v} + \frac{(z + t) p v_{zz}}{u v} +
  \frac{(z + t) u_z^2 v_z}{u^2 v} \\[5mm]
  &+ \frac{(z + t) u_z v_z^2}{u v^2} +\ds \frac{(z + t) q u_z^2}{u^2 v} +
  \frac{((z+t) (u q - v p)-u v) u_z v_z}{u^2 v^2} \\[5mm]
  & -\ds \frac{(z + t) u_z
    q_z}{u v} - \frac{(z + t) p v_z^2}{u v^2} +\frac{(z + t) v_z p_z}{u v} \ds
  - \frac{q(u + (t+z) p) u_z}{u^2 v} \\[5mm]
  &+ \frac{(4 (z+t) u v^3 + z v p - z (t+z) p q) v_z}{z u v^2} + \frac{(z + t)
    (q p_z+p q_z)}{u v}\\[5mm]
  &\ds + \frac{z^2 p q-2 t u v^3}{z^2
    u v}
  \\[5mm]
  (g_0)_{34}&=\ds -2 \frac{(z + t) v_{zzz}}{u} + 2 \frac{(z + t) (q+v_z)
    u_{zz}}{u^2}  - 3 \frac{v_{zz}}{u} + 4
  \frac{(z + t) u_z v_{zz}}{u^2}\\[5mm]
  &\ds  - 2 \frac{(z + t) q_{zz}}{u}- \frac{((t-z) u^2 + 4 (t+z) v^2) u_z^2
    (v_z+q)}{u^3 v^2} +\ds \frac{(z +
    t) u_z v_z^2}{u^2 v}\\[5mm]
  &\ds+ \frac{(2 (t-z) u p + 3 v^2 + 2 (t+z) v q) u_z v_z}{u^2 v^2} + 4
  \frac{(z + t) u_z q_z}{u^2} - \frac{(z + t) p v_z^2}{u^2 v}\\[5mm]
  &\ds+ \frac{((z+t) u^4 v -(z+t) u^2 v^3 +2 z (t-z) u p q + 3 z v^2 q + z
    (z+t) v q^2) u_z}{u^2 v^2 z}\\[5mm]
  &\ds - \frac{((t+z) u^3 v^2 + (t+z) u v^4 + z (t-z) u p^2 + 2 z (z+t) v p
    q) v_z}{u^2 v^2 z} - 3 \frac{q_z}{u}\\[5mm]
  &\ds- \frac{z (z-t) u^4 v p - 2 t u^3 v^3 + z(z+ t) u v^4 q + z^2 (t- z) u
    p^2 q+ z^2 (z+ t) v p q^2}{v^2 z^2 u^2},
  \\[5mm]
  (g_0)_{41}&=\ds (g_0)_{32}+\frac{u_z-p}{v} - \frac{v_z+q}{u} + 2 u v,\\[5mm]
  (g_0)_{43} &= \ds 2 \frac{(z - t) u_{zzz}}{v} - 4 \frac{(z - t) v_z
    u_{zz}}{v^2} + 3 \frac{u_{zz}}{v} + 2 \frac{(z - t) p v_{zz}}{v^2} - 2
  \frac{(z - t) u_z v_{zz}}{v^2}
  \\[5mm]
  &\ds - 2 \frac{(z - t) p_{zz}}{v}- \frac{(z - t) u_z^2 v_z}{u v^2} -
    \frac{(4 (t-z) u^2 + (t+z) v^2) u_z
    v_z^2}{u^2 v^3}-\ds \frac{(z - t) q u_z^2}{u v^2}  \\[5mm]
  &\ds- 2 \frac{(z + t) q u_z
    v_z}{u^2 v}+ 2 \frac{(z - t) p u_z v_z}{u v^2} - 3 \frac{u_z v_z}{v^2} +
    \frac{(4
    (t-z) u^2 + (t+z) v^2) p v_z^2}{u^2 v^3} \\[5mm]
  &+ 4 \frac{(z - t) v_z
    p_z}{v^2}- \frac{(z - t)(u^2 + v^2) u_z}{v z} \ds- \frac{(z + t) q^2
    u_z}{u^2 v} + 2
  \frac{(z - t) p q u_z}{u v^2}
  \\[5mm]
  &+\frac{((t-z) u^2 +(t+z) v^2) v_z}{z u} - \frac{(z - t) p^2 v_z}{u v^2}+ 2
  \frac{(z + t) p q v_z}{u^2 v} + 3 \frac{p v_z}{v^2}\\[5mm]
  & \ds- 3 \frac{p_z}{v} + \frac{(z - t) u^2 p}{v z} - 2 \frac{u t v}{z^2} +
  \frac{(z + t) v^2 q}{z u} - \frac{(z - t) p^2 q}{u v^2} + \frac{(z + t) p
    q^2}{u^2 v},\\[5mm]
  (g_0)_{44}&=\ds\frac{(z - t) v_z u_{zz}}{u v} + \frac{(z - t) q u_{zz}}{u
    v} - \frac{(z - t) p v_{zz}}{u v} + \frac{(z - t) u_z v_{zz}}{uv} \\[5mm]
  &- \frac{(z - t) u_z^2 v_z}{u^2 v} - \frac{(z - t) u_z v_z^2}{u v^2}\ds -
  \frac{(z - t) q u_z^2}{u^2 v} - \frac{(z - t) q u_z v_z}{u v^2} + \frac{u_z
    v_z}{u v} \\[5mm]
  &+ \frac{(z - t) p u_z v_z}{u^2 v} + \frac{(z - t) u_z q_z}{u v}
  + \frac{(z - t) p v_z^2}{u v^2} - \frac{(z - t) v_z p_z}{u v}\\[5mm]
  & \ds + \frac{u_z q}{u v} + \frac{(z - t) p q u_z}{u^2 v} + 4 \frac{(z - t)
    u u_z}{z} - \frac{p v_z}{u v} + \frac{(z - t) p q v_z}{u v^2} \\[5mm]
  &- \frac{(z -
    t) q p_z}{u v} - \frac{(z - t) p q_z}{u v} + 2 \frac{t u^2}{z^2} - \frac{p
    q}{u v}.
\end{align*}

It is readily checked that the operator
$\mathfrak{P}=\sum\limits_{i=0}^3 g_i D_z^i$ is a local Hamiltonian operator
and the flow (\ref{ev}) preserves $\mathfrak{P}$. Moreover, as $\mathfrak{P}$
is linear in $t$, the operators $\mathfrak{P}$ and $\p\mathfrak{P}/\p t$ form
a Hamiltonian pair and are both preserved by the flow (\ref{ev}), cf.\ e.g.\
\cite{asaam, as-sigma} and references therein. Clearly, $\mathfrak{P}$ is the
counterpart of $\mathcal{P}_2$ while $\p\mathfrak{P}/\p t$ is the counterpart
of $\mathcal{P}_1$ for (\ref{ev}).

As $\mathfrak{P}$ explicitly depends on $t$ while the right-hand side of
(\ref{ev}) does not, Lemma 3.8 from \cite{blaszak} implies that there exists
no functional $\mathcal{H}$ such that (\ref{ev}) would be Hamiltonian with
respect to $\mathfrak{P}$ with the Hamiltonian $\mathcal{H}$.  Hence
(\ref{ev}) cannot be written in a bihamiltonian form with respect to the
Hamiltonian pair under study.

On the other hand, (\ref{ev}) is Hamiltonian with respect to
$\partial\mathfrak{P}/\partial t$, that is, we can write~\eqref{ev} in the
form
\[
\boldsymbol{u}_t=\frac{\partial\mathfrak{P}}{\partial t}
\left(\frac{\delta\mathcal{H}}{\delta\boldsymbol{u}}\right),
\]
where $\boldsymbol{u}=(u,v,p,q)^T$, $T$ indicates the transposed matrix, and
$\mathcal{H}=\int h \ dz$,
\[ h={\frac {z u_z^{2 }}{{v}^{2}}}+{\frac {zv_z^{2}}{{u}^{2}}}-2\,{\frac
  {zp{u_z}} {{v}^{2}}}+2\,{\frac {zq{v_z}}{{u}^{2}}}- \left( u-v \right)
\left( v+u \right) +{\frac { \left( {q}^{2}{v}^{ 2}+{u}^{2}{p}^{2} \right)
    z}{{u}^{2}{v}^{2}}}.
\]

\subsection*{Acknowledgments}
It is our great pleasure to thank E.V. Ferapontov for bringing the problem of
study of geometric aspects of integrability for (\ref{sys-eta}) to our
attention and for stimulating discussions.

This research was performed within the framework of and with
financial support of the OPVK program, project
CZ.1.07/2.300/20.0002. AS was also supported in part by the Ministry
of Education, Youth and Sports of the Czech Republic (M\v{S}MT
\v{C}R) under RVO funding for I\v{C}47813059, and by the Grant
Agency of the Czech Republic (GA \v{C}R) under grant P201/12/G028.
IK was also partially supported by the Simons-IUM fellowship.

The authors thank the referee for useful suggestions.


\begin{thebibliography}{99}

\bibitem{blaszak} M. B\l{}aszak, Multi-Hamiltonian theory of dynamical
  systems, Springer, Heidelberg etc., 1998.

\bibitem{b1}O. Bonnet, Th\'eorie des surfaces applicables sur une surface
  donn\'ee, J. de l'\'Ecole Polytech. 24 (1865) 209--230.

\bibitem{b2} O. Bonnet, Memoire sur la th\'eorie des surfaces applicables,
  J. de l'\'Ecole Polytech. 25 (1867) 1-151.

\bibitem{evf} E.V. Ferapontov, Surfaces in 3-space possessing nontrivial
  deformations which preserve the shape operator, in: Proceedings of the
  conference ``Integrable Systems in Differential Geometry'', Tokyo, July
  2000, Contemporary Math. {\bf 508} (202), 145--159, AMS, Providence RI,
  2002, arXiv:math/0107122

\bibitem{evfpc} E.V. Ferapontov, private communication

\bibitem{f} S.P. Finikoff, Surfaces dont les lignes de courbure se
  correspondent avec \'egalit\'e des rayons de courbure principaux homologues,
  C. R. Acad. Sci. 197 (1933) 984--986.\looseness=-1

\bibitem{fg} S.P. Finikoff et B. Gambier, Surfaces dont les lignes de courbure
  se correspondent avec \'egalit\'e des rayons de courbure principaux, Ann.
  Ecole Norm., III. Ser. 50 (1933) 319--370.

\bibitem{l-star} P. Kersten, I.S. Krasil{\cprime}shchik, A. Verbovetsky,
  Hamiltonian operators and $\ell^*$-coverings, J. Geom. Phys. 50 (2004), no. 1,
  273--302, arXiv:math/0304245v5.

\bibitem{kkv}P. Kersten, I.S. Krasil{\cprime}shchik, A. Verbovetsky, R.
  Vitolo, Hamiltonian structures for general PDEs, in: Differential equations:
  geometry, symmetries and integrability, 187--198, Abel Symp., 5, Springer,
  Berlin, 2009, arXiv:0812.4895.

\bibitem{k}I. Krasil{\cprime}shchik, Algebraic theories of brackets and
  related (co)homologies, Acta Appl. Math.  109 (2010), no. 1, 137--150,
  arXiv:0812.4676.

\bibitem{kvv} I.S. Krasil{\cprime}shchik, A. Verbovetsky, R.  Vitolo, A
  unified approach to computation of integrable structures, Acta Appl. Math.
  120 (2012), 199--218, arXiv:1110.4560.

\bibitem{km}J. Krasil{\cprime}shchik, M. Marvan, Coverings and integrability
  of the Gauss--Mainardi--Codazzi equations, Acta Appl. Math.  56 (1999),
  no. 2-3, 217--230, arXiv:solv-int/9812010.

\bibitem{k-v}J. Krasil{\cprime}shchik, A. Verbovetsky, Geometry of jet spaces
  and integrable systems, J. Geom. Phys.  61 (2011), no. 9, 1633--1674,
  arXiv:1002.0077.

\bibitem{kp}M. Kunzinger, R.O. Popovych, Potential conservation laws,
  J. Math. Phys.  49 (2008), no. 10, 103506, 34 pp., arXiv:0803.1156

\bibitem{marvan}M. Marvan, Another look on recursion operators, in:
  Differential geometry and applications (Brno, 1995), 393--402, Masaryk
  Univ., Brno, 1996. Available online at
  \url{http://www.emis.de/proceedings/6ICDGA/IV/}.

\bibitem{m}M. Marvan, Reducibility of zero curvature representations with
  application to recursion operators, Acta Appl. Math.  83 (2004), no. 1-2,
  39--68, arXiv:nlin/0306006.

\bibitem{ms}M. Marvan, A. Sergyeyev, Recursion operator for the stationary
  Nizhnik--Veselov--Novikov equation, J. Phys. A 36 (2003), no. 5, L87--L92,
  arXiv:nlin/0210028.

\bibitem{ms2}M. Marvan, A. Sergyeyev, Recursion operators for dispersionless
  integrable systems in any dimension, Inverse Problems 28 (2012), no. 2,
  025011, 12 pp., arXiv:1107.0784.

\bibitem{m2}O.I. Morozov, A Recursion Operator for the Universal Hierarchy
  Equation via Cartan's Method of Equivalence, Cent. Europ. J. Math. 12
  (2014), no.2, 271--283, arXiv:1205.5748

\bibitem{mo}O.I. Morozov, The four-dimensional Mart\'\i{}nez Alonso--Shabat
  equation: differential coverings and recursion operators, J. Geom. Phys.  85
  (2014), 75--80, arXiv:1309.4993.

\bibitem{mos}O.I. Morozov, A. Sergyeyev, The four-dimensional Mart\'\i{}nez
  Alonso--Shabat equation: reductions and nonlocal symmetries, J. Geom. Phys.
  85 (2014), 40--45, arXiv:1401.7942.

\bibitem{asaam}A. Sergyeyev, A simple way of making a Hamiltonian system into
  a bi-Hamiltonian one, Acta Appl. Math.  83 (2004), no. 1-2, 183--197,
  arXiv:nlin/0310012

\bibitem{asro}A. Sergyeyev, A Simple Construction of Recursion Operators for
  Multidimensional Dispersionless Integrable Systems, arXiv:1501.01955

\bibitem{asbih}A. Sergyeyev, A strange recursion operator demystified,
  J. Phys.  A: Math. Gen. 38 (2005), L257--262, arXiv:nlin/0406032

\bibitem{as-sigma}A. Sergyeyev, Weakly Nonlocal Hamiltonian Structures: Lie
  Derivative and Compatibility, SIGMA 3 (2007), paper 062,
  arXiv:math-ph/0612048
\end{thebibliography}
\end{document}